\documentclass[12pt,reqno]{amsart}
\usepackage{graphicx,amscd, color,amsmath,amsfonts,amssymb,geometry,amssymb, physics}
\usepackage[initials]{amsrefs}

\newcommand{\Ima}{{\operatorname{Im}}}
\newtheorem{Theorem}{Theorem}

\newtheorem{corollary}[Theorem]{Corollary}

\newtheorem{remark}[Theorem]{Remark}

\newtheorem{lemma}[Theorem]{Lemma}

\newtheorem{definition}[Theorem]{Definition}
\numberwithin{equation}{section}

\usepackage{lipsum}

\newcommand\blfootnote[1]{%
  \begingroup
  \renewcommand\thefootnote{}\footnote{#1}%
  \addtocounter{footnote}{-1}%
  \endgroup
}

\usepackage{braket}

\renewcommand\bra[1]{{\langle{#1}|}}
\makeatletter
\renewcommand\ket[1]{%
  \@ifnextchar\bra{\k@t{#1}\!}{\k@t{#1}}%
}
\newcommand\k@t[1]{{|{#1}\rangle}}
\makeatother

\begin{document}

\title{Schmidt rank constraints in Quantum Information Theory}


\author[Cariello]{D. Cariello} 

\blfootnote{\textit{Address.} Faculdade de Matem\'atica,  Universidade Federal de Uberl\^{a}ndia,  38.400-902, Uberl\^{a}ndia, Brazil.\newline\indent  \textit{Email address.} dcariello@ufu.br}




\thanks{ORCID number: 0000-0001-5548-5453}

\begin{abstract}   Can vectors with low Schmidt rank form mutually unbiased bases?
Can vectors with high Schmidt rank form positive under partial transpose states?   In this work,  we address these questions by presenting  several new results related to Schmidt rank constraints and their compatibility with other properties.
We provide an upper bound on the number of mutually unbiased  bases of $\mathbb{C}^m\otimes\mathbb{C}^n$ $(m\leq n)$  formed by vectors with low  Schmidt rank. In particular, the number of mutually unbiased product bases of  $\mathbb{C}^m\otimes\mathbb{C}^n$ cannot exceed $m+1$, which solves a conjecture proposed by McNulty et al. Then  we show how to create a positive under partial transpose entangled  state from any state supported on the antisymmetric space and  how their Schmidt numbers are exactly related.  Finally, we show that the Schmidt number of operator Schmidt rank 3 states of $\mathcal{M}_m\otimes \mathcal{M}_n\ (m\leq n)$ that are invariant under left partial transpose  cannot exceed $m-2$.

\end{abstract}

\maketitle

\section{Introduction}

The Schmidt rank is a fundamental concept  in quantum theory due to its connection to entanglement, so it is  natural to wonder how  constraints on this  rank affect  results related to entanglement and other compatible properties.  In this work, we investigate three situations, relevant to quantum information theory, where  restrictions on this  rank are imposed.
The first   is in the context of mutually unbiased bases. 

Consider $s$ orthonormal bases of a $d-$dimensional  Hilbert space: 
$\{\ket{\Psi_{j1}},\ldots, \ket{\Psi_{jd}}\}$, $j=1,\ldots,s$.

They  are said to be mutually unbiased if   \begin{center}
$|\left\langle\Psi_{aj},\!\Psi_{bi}\right\rangle|= \frac{1}{\sqrt{d}}$
\end{center} 

for every $\{i,j\}\subset \{1,\ldots,d\}$, $\{a,b\}\subset \{1,\ldots,s\}$ and $a\neq b$.  

These bases have been used in state determination, quantum state tomography and cryptography  $($\cite{Ivanovic, wootters, wootters2, cerf, Bennett}$)$.
Determining the maximum number of mutually unbiased bases in an arbitrary dimension $d$ is one of the open problems. It is  known that this number cannot exceed $d+1$. In addition, when $d$ is a prime power, that maximum is exactly $d+1$ $($\cite{wootters, Bandyopadhyay, Calderbank, Weiner, Delsarte,  Asch}$)$.  

Here we address the following problem: \begin{quote}
How big is the number of mutually unbiased bases of $\mathbb{C}^m\otimes\mathbb{C}^n$ or  $\mathbb{R}^m\otimes\mathbb{R}^n$ formed by vectors with Schmidt rank less or equal to $k$, where $k<m\leq n$?
\end{quote}

We show that the number of such bases  cannot exceed
\vspace{0.2cm}

\begin{center}
$\dfrac{k(m^2-1)}{m-k}$ in $\mathbb{C}^m\otimes \mathbb{C}^n$\ \ and\ \ $ \dfrac{k}{2}\dfrac{(m(m+1)-2)}{(m-k)}$ in  $\mathbb{R}^m\otimes \mathbb{R}^n$.
\end{center}

  \vspace{0,2cm}

  Note that our upper bound equals $m+1$, when $k=1$, in $\mathbb{C}^m\otimes \mathbb{C}^n$.  Thus, the number of mutually unbiased product bases of  $\mathbb{C}^m\otimes\mathbb{C}^n$ cannot exceed $m+1$.  This  solves conjecture 1 proposed by McNulty et al. in \cite{Mc} (corollary \ref{corollaryconjecture}).

The connection of these results  with quantum information theory is that they can be interpreted as an upper bound on the number of complementary measurements in bases with little entanglement.

The second situation is in the context of entanglement quantification.  At this point, we are interested in the construction of positive under partial transpose states (PPT states) using vectors with high Schmidt rank. This can be accurately done using the notion of the Schmidt number \cite{Sperling2011, Terhal}.

  Given a positive semidefinite Hermitian matrix  $ \delta=\sum_{i=1}^nA_i\otimes B_i \in \mathcal{M}_k\otimes \mathcal{M}_m$, define  its Schmidt number   by \begin{center}
$\displaystyle SN(\delta)=\min\left\{\max_{ j}\left\{SR(\ket{w_j}) \right\},  \ \delta=\sum_{j=1}^m\ket{w_j}\!\bra{w_j} \right\}$
\end{center}
$($This minimum is taken over all decompositions of $\delta$ as $\sum_{j=1}^m\ket{w_j}\!\bra{w_j}$, where $\ket{w_j}\in\mathbb{C}^k\otimes \mathbb{C}^m$ for every $j$ and $SR(\ket{w_j})$ stands for the Schmidt rank of $\ket{w_j})$.

Recall that $\delta$  is separable if $SN(\delta)=1$ and entangled if  $SN(\delta)>1$ .

A large Schmidt number is associated to an idea of strong entanglement, but entangled PPT states are considered  a weaker form of entanglement. Discovering the best possible Schmidt number for PPT states  has become an important problem (\cite{Yang, Marcus, Chen, sanpera}).

An example of a PPT state with Schmidt number  half of its local dimension has been found recently in \cite[Proposition 2]{CarielloIneq}. This state is a mixture of the orthogonal projection on the symmetric space of $\mathbb{C}^m\otimes \mathbb{C}^m$, which we denote by $P_{sym}^{m,2}$, with a particular pure state.
Although it seems delicate, the construction is actually quite robust. 


Given any state $\gamma$ supported on the antisymmetric subspace of $\mathbb{C}^m\otimes \mathbb{C}^m$, we show  that $$SN(P_{sym}^{m,2}+\epsilon\gamma)=\frac{1}{2}SN(\gamma)$$ and $P_{sym}^{m,2}+\epsilon\gamma$ is positive under partial transpose for  $\epsilon\in\left]0, \frac{1}{6}\right]$ (theorem \ref{theoremPPT}).

Moreover, if $\mathbb{C}^m$  contains $m^2$ equiangular lines (i.e., a  SIC-POVM \cite{Bannai, Bodmann, Scott}) then we can replace $\frac{1}{6}$ above by $\frac{m+1}{6m}$  and the result remains the same. The existence of a SIC-POVM in any $\mathbb{C}^{m}$ is an open  problem. So this little improvement can only be made for some values of $m$ (\cite{Roy}).
But if we know beforehand that $SN(\gamma)>2$ then we can replace $\frac{1}{6}$ by $1$ in the interval above.

These mixtures have been  firstly considered in \cite{sindici} to construct entangled PPT states. It was  already noticed in \cite{sindici} that  $SN(\gamma)>2$ would create an entangled mixture. Later in \cite{CarielloIneq, Pal}, it was noticed that $SN(P_{sym}^{m,2}+\epsilon\gamma)\geq\frac{1}{2}SN(\gamma)$, for any $\epsilon>0$, and  arbitrary state $\gamma$ supported on the antisymmetric space. Our new result  shows how the Schmidt numbers of this PPT mixture and the original $\gamma$ are exactly related for sufficiently small $\epsilon$. 

In the third and final situation, we investigate the relationship between the operator Schmidt rank  and  the Schmidt number of PPT states with some extra conditions.

 The operator Schmidt rank (or tensor rank) of $\delta \in \mathcal{M}_m\otimes \mathcal{M}_n$ is $1$, if $\delta=A_1 \otimes A_2\neq 0$. The operator Schmidt rank of an arbitrary $\gamma \in \mathcal{M}_m\otimes \mathcal{M}_n\setminus\{0\}$ is the minimal number of tensors with operator Schmidt rank 1 that can be added to form $\gamma$.

We show that the Schmidt number of any state of  $\mathcal{M}_m\otimes \mathcal{M}_n\ (m \leq n)$   invariant under left partial transpose with operator Schmidt rank 3 is at most $m-2$ (corollary \ref{corollaryinvariant}).  It complements  \cite[Theorem 5]{Marcus}. 
This result is also related  to the conjecture that says that the Schmidt number of any PPT state of  $\mathcal{M}_m\otimes \mathcal{M}_n\ (m \leq n)$ cannot  be $m$ (\cite{sanpera}).

In particular, this theorem says that   every state invariant under left partial transpose with operator Schmidt rank 3 in $\mathcal{M}_3\otimes \mathcal{M}_n$ is separable (theorem \ref{lowtensorrank}).
This result is a new contribution to an ongoing investigation that relates low operator Schmidt rank  to separability.

  States with operator Schmidt rank 2 are always separable (See \cite[Theorem 58]{cariello_QIC1} or \cite{gemma}). In addition,  states  of $\mathcal{M}_2\otimes \mathcal{M}_m$ with operator Schmidt rank 3 are also separable (See \cite[Theorem 19]{Cariello_QIC}). However, this is not valid in  $\mathcal{M}_3\otimes \mathcal{M}_m\ (m\geq 3)$ (See \cite[Proposition 25]{Cariello_QIC}). The invariance under left partial transpose is a sufficient condition for the separability of  states of  $\mathcal{M}_3\otimes \mathcal{M}_m$ with operator Schmidt rank 3.

This work is organized as follows.  In section 2, we obtain an upper bound on the number of mutually unbiased based formed by vectors with Schmidt rank less or equal to $k$. In section 3, we constructed entangled PPT states from  states supported on the antisymmetric subspace and we show how their Schmidt numbers are exactly related. In section 4, we prove that the Schmidt number of operator Schmidt rank 3 states of $\mathcal{M}_k\otimes \mathcal{M}_m\ (k\leq m)$ that are invariant under left partial transpose  cannot exceed $k-2$.

\section{Mutually unbiased  bases }

In this section we provide an upper bound on the number of mutually unbiased  bases of  $\mathbb{C}^m\otimes \mathbb{C}^n$ formed  by vectors with Schmidt rank less or equal to $k$ $(k<m\leq n)$. \\

Denote by $\mathcal{M}_k$ the set of complex matrices of order $k$. Identify  $\mathcal{M}_k\otimes \mathcal{M}_m\simeq \mathcal{M}_{km}$ and  $\mathbb{C}^k\otimes \mathbb{C}^m\simeq \mathbb{C}^{km}$ via Kronecker product.  Let $F_d\in \mathcal{M}_d\otimes \mathcal{M}_d$ be the flip operator $($i.e., $F_d(\ket{a}\otimes \ket{b})=\ket{b}\otimes \ket{a},$ for every $\ket{a},\ket{b}\in\mathbb{C}^d)$.
  
  \vspace{0,3cm}

\begin{definition}\label{definition1}  Let $\mathcal{P}(\rho)=Tr(\rho^2)$, where $\rho$ is a square matrix and $Tr(\rho)$ is its trace.
Denote  the left and the right partial trace of $\gamma\in \mathcal{M}_m\otimes \mathcal{M}_n$ by  $Tr_A(\gamma)\in M_n$ and $Tr_B(\gamma)\in M_m$, respectively.
 Let  the Schmidt rank of $\ket{w}\in\mathbb{C}^m\otimes \mathbb{C}^n$ be the rank of $Tr_A(\ket{w}\!\bra{w})$ and denote it by  $SR(\ket{w})$.\\
\end{definition}
\begin{remark}\label{remarkeasy}Let $Y=Tr_A(\ket{w}\!\bra{w})$. By Cauchy-Schwarz inequality, notice that $$\mathcal{P}(Tr_A(\ket{w}\!\bra{w}))=Tr(Y^2)\geq \frac{Tr(Y)^2}{\text{rank }(Y)}.$$

If $\ket{w}$ is a unit vector then $Tr(Y)=1$. Since $\text{rank }(Y)=\text{SR}(\ket{w})$, $\mathcal{P}(Tr_A(\ket{w}\!\bra{w}))\geq\dfrac{1}{\text{SR}(\ket{w})}.$

This last inequality  shall be used in corollary \ref{maincorollary}.
\end{remark}

\vspace{0,3cm}

\begin{definition}\label{definition3}  Let $\{\ket{e_1},\ldots,\ket{e_m}\}$ and $\{\ket{f_1},\ldots \ket{f_n}\}$ be the canonical bases of $\mathbb{C}^m$ and $\mathbb{C}^n$, respectively. 
 \begin{enumerate}
 \item Let $\displaystyle \ket{\Phi}=\sum_{j=1}^n\ket{f_j}\otimes\ket{f_j}\in \mathbb{C}^n\otimes \mathbb{C}^n$.
 \item Let  $\displaystyle \ket{\Psi}=\sum_{i=1}^m\sum_{j=1}^n\ket{e_i}\otimes\ket{f_j}\otimes \ket{e_i}\otimes\ket{f_j}\in \mathbb{C}^m\otimes\mathbb{C}^n\otimes \mathbb{C}^m\otimes \mathbb{C}^n$.
\item Let $Tr_{1,3}(X)\in \mathcal{M}_n\otimes\mathcal{M}_n$ be the partial trace of  $X\in\mathcal{M}_m\otimes\mathcal{M}_n\otimes\mathcal{M}_m\otimes \mathcal{M}_n$ $($We are tracing out the first and the third sites$)$.\\
\item  Let  the functional $f: \mathcal{M}_m\otimes\mathcal{M}_n\otimes\mathcal{M}_m\otimes \mathcal{M}_n\rightarrow \mathbb{C}$ be as  $f(X)=Tr(Tr_{1,3}(X)\ket{\Phi}\!\bra{\Phi}).$

    Note that $f$ is a positive functional, i.e., it sends positive semidefinite Hermitian matrices to non-negative  real numbers. \\
 
 \end{enumerate}
\end{definition}

\begin{lemma}\label{mainlemma} Let $\ket{\omega}\in \mathbb{C}^m\otimes \mathbb{C}^n$. Then \\
 \begin{enumerate}
\item $f(\ket{\omega}\!\bra{\omega}\otimes \ket{\overline{\omega}}\!\bra{\overline{\omega}})=\mathcal{P}(Tr_A(\ket{\omega}\!\bra{\omega}))$,\\

\item $f(\ket{\Psi}\!\bra{\Psi})=mn^2$,\\

\item $f(Id_{m\times m}\otimes Id_{n\times n}\otimes Id_{m\times m}\otimes Id_{n\times n})=m^2n$.\\

\item $f(F_{mn})=mn$, 

\vspace{0.2cm}

\noindent where $F_{mn}\in \mathcal{M}_{mn}\otimes \mathcal{M}_{mn}$ is the flip operator. Recall the identification $M_m\otimes M_n\simeq M_{mn}.$

\item $f(P_{sym}^{mn,2})=\dfrac{m^2n+mn}{2}$, 

\vspace{0.2cm}

\noindent where $P_{sym}^{mn,2}\in \mathcal{M}_{mn}\otimes \mathcal{M}_{mn}$ is the orthogonal projection on the symmetric subspace of $\mathbb{C}^{mn}\otimes \mathbb{C}^{mn}$.
Recall that $P_{sym}^{mn,2}=\frac{1}{2}(Id_{m\times m}\otimes Id_{n\times n}\otimes Id_{m\times m}\otimes Id_{n\times n}+F_{mn}).$


\end{enumerate}
\end{lemma}
\begin{proof}
The proof of this lemma is straightforward. It is left to the reader. 
\end{proof}

\vspace{0.5cm}

\begin{Theorem}\label{maintheorem} Let $\{\ket{\Psi_{j1}},\ldots, \ket{\Psi_{j(mn)}}\}$,  $j = 1,\ldots,t$, be mutually unbiased bases of a $mn$-dimensional Hilbert space $\mathcal{H}$.   Then
\begin{itemize} 
\item[$a)$] $\displaystyle \sum_{j=1}^t\sum_{i=1}^{mn} \mathcal{P}(Tr_A(\ket{\Psi_{ji}}\!\bra{\Psi_{ji}}))\leq (m^2+t-1)n$, if $\mathcal{H}=\mathbb{C}^m \otimes
 \mathbb{C}^n$,
\item[$b)$] $\displaystyle\sum_{j=1}^t\sum_{i=1}^{mn} \mathcal{P}(Tr_A(\ket{\Psi_{ji}}\!\bra{\Psi_{ji}}))\leq  \left(\dfrac{m(m+1)}{2}+t-1\right)n,$ if $\mathcal{H}=\mathbb{R}^m \otimes
 \mathbb{R}^n$.
\end{itemize} 

\end{Theorem}
\begin{proof} 
Consider the orthogonal projections $A_{1},\ldots, A_{t}\in \mathcal{M}_m\otimes\mathcal{M}_n\otimes\mathcal{M}_m\otimes \mathcal{M}_n$ defined by 
$$\displaystyle A_{j}=\sum_{i=1}^{d}\ket{\Psi_{ji}}\bra{\Psi_{ji}}\otimes \ket{\overline{\Psi_{ji}}}\!\bra{\overline{\Psi_{ji}}}.$$

By  \cite[Lemma 34]{CarielloIEEE}, $A_{j}A_{k}=A_{k}A_{j}=\dfrac{1}{mn}\ket{\Psi}\!\bra{\Psi}$, for every $j,k\in\{1,\dots,t\}$ and $j\neq k$.\\

Therefore, the matrix  

\begin{equation}\label{eqA}
 B =Id_{m\times m}\otimes Id_{n\times n}\otimes Id_{m\times m}\otimes Id_{n\times n}+\dfrac{t-1}{mn}\ket{\Psi}\!\bra{\Psi}-\sum_{j=1}^tA_{j}
\end{equation} 

is positive semidefinite.

\vspace{0.5cm}

 \underline{Case $a)$: $\mathcal{H}=\mathbb{C}^m \otimes
 \mathbb{C}^n$}\\

 By lemma \ref{mainlemma}, equation (\ref{eqA}) and the positivity of $B$, we have 

$$ f(B)=m^2n+(t-1)n-\sum_{j=1}^t\sum_{i=1}^{mn} \mathcal{P}(Tr_A(\ket{\Psi_{ji}}\!\bra{\Psi_{ji}}))\geq 0.$$

Finally,\ \  $\displaystyle \sum_{j=1}^t\sum_{i=1}^{mn} \mathcal{P}(Tr_A(\ket{\Psi_{ji}}\!\bra{\Psi_{ji}}))\leq (m^2+t-1)n.$\\\\

 \underline{Case $b)$: $\mathcal{H}=\mathbb{R}^m \otimes
 \mathbb{R}^n$}\\

In this case, $\ket{\overline{\Psi_{ji}}}=\ket{\Psi_{ji}}$ for every $j,i$. Therefore every $A_j$ is supported on the symmetric subspace of $\mathcal{H}\otimes\mathcal{H}$. Hence,

\begin{equation}
\label{eqD} \displaystyle P_{sym}^{mn,2}BP_{sym}^{mn,2} =P_{sym}^{mn,2}+\dfrac{t-1}{mn}\ket{\Psi}\!\bra{\Psi}-\sum_{j=1}^tA_j
\end{equation}

 is positive semidefinite.

 By lemma \ref{mainlemma}, equation (\ref{eqD}) and the positivity of $P_{sym}^{mn,2}BP_{sym}^{mn,2}$, we have 
$$ f(P_{sym}^{mn,2}BP_{sym}^{mn,2})= \left(\dfrac{m(m+1)}{2}+t-1\right)n-\sum_{j=1}^t\sum_{i=1}^{mn} \mathcal{P}(Tr_A(\ket{\Psi_{ji}}\!\bra{\Psi_{ji}}))\geq 0.
 $$
 
Finally,\ \ \  $\displaystyle \sum_{j=1}^t\sum_{i=1}^{mn} \mathcal{P}(Tr_A(\ket{\Psi_{ji}}\!\bra{\Psi_{ji}}))\leq \left(\dfrac{m(m+1)}{2}+t-1\right)n.$
\end{proof}

\vspace{0.3cm}

\begin{remark}\label{remarkconservationlaws}  Using the notation of the proof of Theorem \ref{maintheorem}, we can describe the equation obtained in \cite[Lemma 35]{CarielloIEEE} as $\displaystyle B=Id_{m\times m}\otimes Id_{n\times n}\otimes Id_{m\times m}\otimes Id_{n\times n}+\ket{\Psi}\!\bra{\Psi}-\sum_{j=1}^{mn+1}A_{j}=0.$ Thus,
  \begin{center}
$\displaystyle 0=f(B)=mn(m+n)-\sum_{j=1}^{mn+1}\sum_{i=1}^{mn} \mathcal{P}(Tr_A (\ket{\Psi_{ji}}\!\bra{\Psi_{ji}})).$ 
\end{center}
Therefore,  we recover the conservation law obtained in \cite{Wiesniak}.
\end{remark}

\vspace{0,5cm}

\begin{corollary}\label{maincorollary} Let $k<m\leq n $. The  number of mutually unbiased bases  formed by vectors with Schmidt rank less or equal to $ k$ cannot exceed 

\begin{itemize}
\item[$a)$] $\dfrac{k(m^2-1)}{m-k}$ in $\mathbb{C}^m\otimes\mathbb{C}^n$ and
\item[$b)$] $ \dfrac{k}{2}\dfrac{(m(m+1)-2)}{(m-k)}$ in $\mathbb{R}^m\otimes\mathbb{R}^n$.
\end{itemize}
 In particular, the  number of mutually unbiased product bases  cannot exceed $ m+1$ in  $\mathbb{C}^m\otimes\mathbb{C}^n$.
\end{corollary} 
\begin{proof}
 Let $\{\ket{\Psi_{j1}},\ldots, \ket{\Psi_{j(mn)}}\}$, for $j = 1,\ldots,t$, be mutually unbiased bases formed by vectors with Schmidt rank less or equal to $k$. \\
 
Since the Schmidt rank  of each unit vector  $ \ket{\Psi_{ji}}$ is less or equal to $k$,\ $\mathcal{P}(Tr_A(\ket{\Psi_{ji}}\!\bra{\Psi_{ji}}))\ge\dfrac{1}{k}$ (Remark \ref{remarkeasy}).\\
 
By  Theorem \ref{maintheorem},  
\begin{itemize}
\item[$a)$]  $\displaystyle t\ \dfrac{mn}{k}\leq \sum_{j=1}^{t}\sum_{i=1}^{mn} \mathcal{P}(Tr_A(\ket{\Psi_{ji}}\!\bra{\Psi_{ji}}))\leq (m^2+t-1)n$  in  $\mathbb{C}^m\otimes\mathbb{C}^n$ and
\item[$b)$]  $\displaystyle t\ \dfrac{mn}{k}\leq \sum_{j=1}^{t}\sum_{i=1}^{mn} \mathcal{P}(Tr_A(\ket{\Psi_{ji}}\!\bra{\Psi_{ji}}))\leq \left(\dfrac{m(m+1)}{2}+t-1\right)n$ in    $\mathbb{R}^m\otimes\mathbb{R}^n$.
\end{itemize}

\vspace{0.3cm}

Hence,     
\begin{itemize}
\item[$a)$]  $t\leq k\left(\dfrac{m^2-1}{m-k}\right)$ in  $\mathbb{C}^m\otimes\mathbb{C}^n$ and
\item[$b)$]  $t\leq \dfrac{k}{2}\dfrac{(m(m+1)-2)}{(m-k)}$  in  $\mathbb{R}^m\otimes\mathbb{R}^n$.
\end{itemize}
\end{proof}

\vspace{0,3cm}

\begin{remark}The upper bounds obtained in the last corollary do not depend on $n$. For instance, in $\mathbb{C}^m\otimes \mathbb{C}^n$  for $n$ much larger than $m$,  our upper bound  turns out to be much smaller than  $mn+1$. In fact, we have  \begin{center}
 $\dfrac{k(m^2-1)}{m-k}<mn+1$, if and only if, $k< \dfrac{mn+1}{m+n}$.
 \end{center}
 
Now, if no restriction on $n$ is imposed, besides $n\geq m$,   we have
$\dfrac{k(m^2-1)}{m-k}<mn-1$ for $k\leq \dfrac{m}{2}$.

This is  interesting because   $mn-1$ turns out to be  an upper bound on the number of mutually unbiased bases of  $\mathbb{C}^m\otimes\mathbb{C}^n$ formed by vectors with   Schmidt coefficients equal to $\dfrac{1}{\sqrt{k}}$ and  $n$ is a multiple of $m$ \cite[Theorem 4]{Shi}. Our result improves this upper bound when $k\leq \dfrac{m}{2}$.
Nothing can be said about the case $k>\dfrac{m}{2}$ with our method.
There is an extensive literature on these bases with fixed Schmidt coefficients \cite{XU, XU2, Guo, Mao, Werner}.  

\end{remark}

\vspace{0,3cm} 

The next corollary solves conjecture 1 in \cite{Mc}.
\vspace{0,3cm} 

\begin{corollary}\label{corollaryconjecture} The maximum number of mutually unbiased product bases of $\mathbb{C}^{d_1}\otimes\ldots\otimes\mathbb{C}^{d_n}$ is less or equal to $\displaystyle \min_{j} d_j+1$. Note that if $d_1,\ldots,d_n$ are powers of distinct primes then this maximum number is exactly $\displaystyle\min_{j} d_j+1$.
\end{corollary}

\begin{proof}
 Assume without loss of generality that   $d_1+1=\displaystyle\min_{j} d_j+1$. Since a product vector in $\mathbb{C}^{d_1}\otimes\ldots\otimes\mathbb{C}^{d_n}$ is also a product vector in $\mathbb{C}^{d_1}\otimes\mathbb{C}^{d_2\ldots d_n}$, the maximum number of mutually unbiased product bases of $\mathbb{C}^{d_1}\otimes\ldots\otimes\mathbb{C}^{d_n}$ is less or equal to the same number in  $\mathbb{C}^{d_1}\otimes\mathbb{C}^{d_2\ldots d_n}$.
 
  By corollary \ref{maincorollary}, the maximum number of mutually unbiased product bases of $\mathbb{C}^{d_1}\otimes\mathbb{C}^{d_2\ldots d_n}$ cannot exceed $d_1+1$. 
\end{proof}

\section{Entangled PPT Mixtures}

Let us call $\delta\in \mathcal{M}_k\otimes \mathcal{M}_m$ a   state, if  $\delta$ is a positive semidefinite Hermitian matrix with trace 1.

In this section, we show that  $P_{sym}^{k,2}+\epsilon\gamma$ is  positive under partial transpose and 
 $$SN(P_{sym}^{k,2}+\epsilon\gamma)=\frac{SN(\gamma)}{2}$$
 
 for sufficiently small $\epsilon$, where  $\gamma\in \mathcal{M}_k\otimes \mathcal{M}_k$ is any state supported on the antisymmetric subspace of $\mathbb{C}^k\otimes\mathbb{C}^k$ (theorem \ref{theoremPPT}).  In order to obtain this result, we need the following equation obtained in \cite{Klappenecker}.

If $\{\ket{\Psi_{j1}},\ldots, \ket{\Psi_{jk}}\}$, $1\leq j\leq  k+1$, are $k+1$ mutually unbiased bases of $\mathbb{C}^k$ then 

\begin{equation}\label{Conservation2}
2P_{sym}^{k,2}=\sum_{j=1}^{k+1}\sum_{i=1}^{k}\ket{\Psi_{ji}}\!\bra{\Psi_{ji}}\otimes \ket{\Psi_{ji}}\!\bra{\Psi_{ji}}\in \mathcal{M}_k\otimes \mathcal{M}_k.
\end{equation}

\vspace{0.3 cm}

\begin{definition}\label{definition10}
Let the right partial transpose  of  $\displaystyle \delta=\sum_{i=1}^nA_i\otimes B_i \in \mathcal{M}_k\otimes \mathcal{M}_m$ be $\displaystyle \delta^{\Gamma}=\sum_{i=1}^nA_i\otimes B_i^t$
$($The left partial transpose is defined analogously$)$.
Moreover, let us say that  $\delta$ is positive under partial transpose or simply  PPT   if  $\delta$ and $\delta^{\Gamma}$ are positive semidefinite Hermitian matrices.
 \end{definition}

\vspace{0.1cm}

\begin{lemma}\label{lemmasr=1} Let  $\ket{a}\in \mathbb{C}^k$ be a unit vector. Then $P_{sym}^{k,2}-\epsilon \ket{a}\!\bra{a}\otimes \ket{a}\!\bra{a}\in \mathcal{M}_k\otimes \mathcal{M}_k$ is separable for $\epsilon\leq \frac{1}{2}$. In addition, if $\mathbb{C}^k$ contains a SIC-POVM then the same matrix is separable if $\epsilon\leq \frac{k+1}{2k}$. 
\end{lemma}  
\begin{proof}
Let $n$ be a prime number greater than $k$. Let $\{\ket{a_{j1}},\ldots, \ket{a_{jn}}\}$, $1\leq j\leq  n+1$, be $n+1$ mutually unbiased bases of $\mathbb{C}^n$ \cite{Ivanovic}. We can assume without loss of generality that  $\ket{a_{11}}=\left(\begin{array}{c}
  \ket{a} \\ 
  0
  \end{array}\right) \in \mathbb{C}^k\times\mathbb{C}^{n-k}$.
  
 By   equation \ref{Conservation2},  
$$P_{sym}^{n,2}=\frac{1}{2}\left(\sum_{j=1}^{n+1}\sum_{i=1}^{n}\ket{a_{ji}}\!\bra{a_{ji}}\otimes \ket{a_{ji}}\!\bra{a_{ji}}\right).$$

Thus,  \ \  $B=P_{sym}^{n,2}- \epsilon \ket{a_{11}}\!\bra{a_{11}}\otimes\ket{a_{11}}\!\bra{a_{11}}$

$$ = \left(\frac{1}{2}- \epsilon\right) \ket{a_{11}}\!\bra{a_{11}}\otimes\ket{a_{11}}\!\bra{a_{11}}+\frac{1}{2}\left(\sum_{i=2}^n \ket{a_{1i}}\!\bra{a_{1i}}\otimes\ket{a_{1i}}\!\bra{a_{1i}}+
\sum_{j=2}^{n+1}\sum_{i=1}^{n}\ket{a_{ji}}\!\bra{a_{ji}}\otimes \ket{a_{ji}}\!\bra{a_{ji}}\right).$$

 is  separable for $\epsilon\leq \frac{1}{2}$.\\

 \vspace{0.1cm}

Now, let $U_{k\times n}=(Id_{k\times k}\ 0_{k\times n-k})$ and  note that

$$(U\otimes U)B(U^*\otimes U^*) =P_{sym}^{k,2}- \epsilon \ket{a}\!\bra{a}\otimes \ket{a}\!\bra{a}.$$

 \vspace{0.3cm}

  So $P_{sym}^{k,2}- \epsilon \ket{a}\!\bra{a}\otimes \ket{a}\!\bra{a}\in \mathcal{M}_k\otimes \mathcal{M}_k$ is separable too for $\epsilon\leq \frac{1}{2}$.\\
  
\vspace{0.1cm}

Next, if  $\mathbb{C}^k$ contains a SIC-POVM then we can write $$P_{sym}^{k,2}=\sum_{i=1}^{k^2}\frac{k+1}{2k} \ket{v_{i}}\!\bra{v_i}\otimes\ket{v_{i}}\!\bra{v_i}\, $$ 

where $\ket{v_1}=\ket{a}$ (\cite[Definition 2.1]{Roy}). Thus, $P_{sym}^{k,2}- \epsilon \ket{a}\!\bra{a}\otimes \ket{a}\!\bra{a}=$ 
$$\left(\frac{k+1}{2k}-\epsilon\right) \ket{v_{1}}\!\bra{v_1}\otimes\ket{v_{1}}\!\bra{v_1}+\sum_{i=2}^{k^2}\frac{k+1}{2k}  \ket{v_{i}}\!\bra{v_i}\otimes\ket{v_{i}}\!\bra{v_i}.$$ 
  
In this case, $P_{sym}^{k,2}- \epsilon \ket{a}\!\bra{a}\otimes \ket{a}\!\bra{a}$ is separable for 
 $\epsilon\leq \frac{k+1}{2k}$. 
\end{proof}  

\vspace{0.3cm}

  \begin{lemma} \label{lemmasr=2}Let $\ket{a_1},\ket{a_2}$ be orthonormal vectors of $\mathbb{C}^k$ and $\ket{s}=\ket{a_1}\otimes \ket{a_2}+\ket{a_2}\otimes \ket{a_1}$. Consider $B=  P_{sym}^{k,2}-\epsilon \ket{s}\!\bra{s}\in \mathcal{M}_k\otimes \mathcal{M}_k$. Then 
  
\begin{itemize}
\item[$a)$]  $SN(B)\leq 2$ for $\epsilon\leq \frac{1}{2}$ and arbitrary $k$,
\item[$b)$] $SN(B)=1$  for $\epsilon\in \left[0,\frac{1}{12}\right]$ and arbitrary $k$, 
\item[$c)$] $SN(B)=1$ for $\epsilon\in \left[0,\frac{k+1}{12k}\right]$, if  $\mathbb{C}^k$ contains a SIC-POVM.

\end{itemize}

\end{lemma}  
\begin{proof} 
 \underline{Part $a)$:} Let $\ket{a_1},\ldots,\ket{a_k}$ be an orthonormal basis of $\mathbb{C}^k$. Since  $P_{sym}^{k,2}$ is the projection on the symmetric subspace of $\mathbb{C}^k\otimes \mathbb{C}^k$,
$$\displaystyle P_{sym}^{k,2}=\sum_{i=1}^k\ket{a_i}\!\bra{a_i}\otimes\ket{a_i}\!\bra{a_i}+\sum_{1\leq i<j\leq k}\frac{1}{2}(\ket{a_i}\otimes \ket{a_j}+\ket{a_j}\otimes \ket{a_i})(\bra{a_i}\otimes \bra{a_j} +\bra{a_j}\otimes \bra{a_i}).$$

Hence, $B=\displaystyle P_{sym}^{k,2}-\epsilon \ket{s}\!\bra{s}=\sum_{i=1}^k\ket{a_i}\!\bra{a_i}\otimes\ket{a_i}\!\bra{a_i}+\left(\frac{1}{2}-\epsilon\right) \ket{s}\!\bra{s}+$
$$\hspace{5,5cm}+\sum_{\stackrel{1\leq i<j\leq k}{(i,j)\neq (1,2)}}\frac{1}{2} (\ket{a_i}\otimes \ket{a_j}+\ket{a_j}\otimes \ket{a_i})(\bra{a_i}\otimes \bra{a_j} +\bra{a_j}\otimes \bra{a_i}).$$

Thus, $SN(B)\leq 2$ for $\epsilon\leq \frac{1}{2}$.

\vspace{1cm}

 \underline{Parts $b)$ and $c)$:}  Let $\{\ket{e_1},\ket{e_2}\}$, $\{\ket{v_1},\ket{v_2}\}$ and $\{\ket{w_1},\ket{w_2}\}$ be 3 mutually unbiased bases of $\mathbb{C}^2$, where   $\{\ket{e_1},\ket{e_2}\}$ is the canonical basis \cite{Ivanovic}.\\
  
By equation \ref{Conservation2},
$$
2P_{sym}^{2,2}=
 \sum_{i=1}^2 \ket{e_i}\!\bra{e_i}\otimes \ket{e_i}\!\bra{e_i}+
 \sum_{i=1}^2 \ket{v_i}\!\bra{v_i}\otimes \ket{v_i}\!\bra{v_i}+\sum_{i=1}^2 \ket{w_i}\!\bra{w_i}\otimes \ket{w_i}\!\bra{w_i}.
$$

\vspace{0.3cm}

Moreover, since  $P_{sym}^{2,2}$ is the projection on the symmetric subspace of $\mathbb{C}^2\otimes \mathbb{C}^2$, $$
2P_{sym}^{2,2}=\sum_{i=1}^2  2(\ket{e_i}\!\bra{e_i}\otimes \ket{e_i}\!\bra{e_i})+\ket{v}\!\bra{v},
 $$

where $\ket{v}=\ket{e_1}\otimes \ket{e_2}+\ket{e_2}\otimes \ket{e_1}$.\\
 
 \vspace{0.3cm}

Next, define the isometry  $U_{k\times 2}$  as $U\ket{e_1}=\ket{a_1}$ and $U\ket{e_2}=\ket{a_2}$.\\ 

\vspace{0.3cm}
 
Note that $(U\otimes U)(2P_{sym}^{2,2})(U^*\otimes U^*)= $

 $$
 = \sum_{j=1}^22\ket{a_{j}}\!\bra{a_j}\otimes \ket{a_j}\!\bra{a_j}+\ket{s}\!\bra{s}=\sum_{i=1}^6 \ket{b_i}\!\bra{b_i}\otimes \ket{b_i}\!\bra{b_i},
 $$

where $\ket{b_1}=\ket{a_1},\ \ket{b_2}=\ket{a_2},\ \ket{b_3}= U \ket{v_1}, \ket{b_4}= U\ket{v_2},\ \ket{b_5}= U\ket{w_1},\ \ket{b_6}= U\ket{w_2}$.\\

\vspace{0.3cm}

In addition,  $\ket{b_1},\ldots,\ket{b_6}$ are unit vectors, since $U$ is an isometry. \\
 
 Thus,  $
 \displaystyle 2P_{sym}^{k,2}- \epsilon\left(\sum_{j=1}^22\ket{a_{j}}\!\bra{a_j}\otimes \ket{a_j}\!\bra{a_j}+\ket{s}\!\bra{s}\right)=$ $$=2P_{sym}^{k,2}- \epsilon\left(\sum_{i=1}^6 \ket{b_i}\!\bra{b_i}\otimes \ket{b_i}\!\bra{b_i}\right)
 $$
 
  $$
  =\sum_{i=1}^6\frac{1}{6}\left(2P_{sym}^{k,2}- 6\epsilon\ \ket{b_i}\!\bra{b_i}\otimes \ket{b_i}\!\bra{b_i}\right)$$
  
  $$=\sum_{i=1}^6\frac{1}{3}\left(P_{sym}^{k,2}- 3\epsilon\ \ket{b_i}\!\bra{b_i}\otimes \ket{b_i}\!\bra{b_i}\right),
 $$
  is  separable for $\epsilon\in \left[0,\frac{1}{6}\right]$, when  $k$ is arbitrary, or for $\epsilon\in \left[0,\frac{k+1}{6k}\right]$, when $\mathbb{C}^k$ contains a SIC-POVM  by lemma \ref{lemmasr=1}.  \\

Finally, $P_{sym}^{k,2}-\epsilon \ket{s}\!\bra{s}=\frac{1}{2}(2P_{sym}^{k,2}-2\epsilon \ket{s}\!\bra{s})=$

$$\frac{1}{2} \left[ 2P_{sym}^{k,2}- 2\epsilon\left(\sum_{j=1}^22\ket{a_{j}}\!\bra{a_j}\otimes \ket{a_j}\!\bra{a_j}+\ket{s}\!\bra{s}\right)\right]+\epsilon\left(\sum_{j=1}^22\ket{a_{j}}\!\bra{a_j}\otimes \ket{a_j}\!\bra{a_j}\right).$$

We have just noticed that the first summand above  is separable for  $2\epsilon\in \left[0,\frac{1}{6}\right]$, when  $k$ is arbitrary, or for $2\epsilon\in \left[0,\frac{k+1}{6k}\right]$, when $\mathbb{C}^k$ contains a SIC-POVM.    \vspace{0.2cm}

Therefore, $P_{sym}^{k,2}-\epsilon \ket{s}\!\bra{s}$ is separable for  $\epsilon\in \left[0,\frac{1}{12}\right]$, when  $k$ is arbitrary, or for $\epsilon\in \left[0,\frac{k+1}{12k}\right]$, when $\mathbb{C}^k$ contains a SIC-POVM.

\end{proof}

  \vspace{0.3cm}
   \begin{lemma} \label{lemmageneral}Let  $\ket{v}$ be a unit antisymmetric vector of  $\mathbb{C}^k\otimes \mathbb{C}^k $. Consider $B=P_{sym}^{k,2}+\epsilon \ket{v}\!\bra{v}\in M_k\otimes M_k$. Then   $B$ is PPT and
   
\begin{itemize}
\item[$a)$] $SN(B)\leq \max\left\{\dfrac{SR(\ket{v})}{2},2\right\}$ for $\epsilon\in [0,1]$ and arbitrary k, 
\item[$b)$]  $SN(B)\leq\dfrac{SR(\ket{v})}{2}$  for  $\epsilon\in\left [0,\frac{1}{6}\right]$ and arbitrary k, 
\item[$c)$]   $SN(B)\leq \dfrac{SR(\ket{v})}{2}$  for  $\epsilon\in\left [0,\frac{k+1}{6k}\right]$, if $\mathbb{C}^k$ contains a SIC-POVM .
\end{itemize}

\end{lemma}  
\begin{proof} Let $SR(\ket{v})=2n$. By   \cite[Corollary 4.4.19.]{Horn}, there are positive numbers  $\lambda_1,\ldots,\lambda_n$  and orthonormal vectors $\ket{v_1},\ldots,\ket{v_n}, \ket{w_1},\ldots,\ket{w_n} $  of $\mathbb{C}^k$ and such that  \begin{center}
$\displaystyle\ket{v}=\sum_{i=1}^n\lambda_i (\ket{v_i}\otimes \ket{w_i}-\ket{w_i}\otimes \ket{v_i})$ and $\displaystyle 2(\sum_{i=1}^n\lambda_i^2)=1$. 
\end{center}
  
 Define $\ket{m_i}=\lambda_i (\ket{v_i}\otimes \ket{w_i}+\ket{w_i}\otimes \ket{v_i})$ for $i=1,\ldots, n$.\\

 \vspace{0,3cm}

By induction on $n$, we can easily show that 
$$\displaystyle P_{sym}^{k,2}+\epsilon \ket{v}\!\bra{v}= P_{sym}^{k,2}- \epsilon \ket{m_1}\!\bra{m_1} -  \ldots - \epsilon \ket{m_n}\!\bra{m_n}+\sum_{i_1,\ldots,i_n=1}^2 \dfrac{\epsilon}{2^n}\ket{v_{i_1,\ldots,i_n}}\!\bra{v_{i_1,\ldots,i_n}},
$$

where $\ket{v_{i_1,\ldots,i_n}}=\ket{m_n}+(-1)^{i_1}\ket{m_{n-1}}+ \ldots +(-1)^{i_{n-1}}\ket{m_1} +(-1)^{i_n}\ket{v}$.\\

\vspace{0,3cm}

 Hence, \begin{equation}\label{eqSN}B=P_{sym}^{k,2}+\epsilon \ket{v}\!\bra{v}=\displaystyle \sum_{i=1}^n 2\lambda_i^2\left( P_{sym}^{k,2} -\dfrac{\epsilon}{2} \dfrac{\ket{m_i}\!\bra{m_i}}{\lambda_i^2}\right)+ \sum_{i_1,\ldots,i_n=1}^2 \dfrac{\epsilon}{2^n}\ket{v_{i_1,\ldots,i_n}}\!\bra{v_{i_1,\ldots,i_n}}.
\end{equation}

 \vspace{0,3cm}

Next, by lemma \ref{lemmasr=2},
\begin{itemize}
\item[$a)$] $SN\left( P_{sym}^{k,2} -\dfrac{\epsilon}{2} \dfrac{\ket{m_i}\!\bra{m_i}}{\lambda_i^2}\right)\leq 2$
, when  $k$ is arbitrary and $\frac{\epsilon}{2}\in\left [0,\frac{1}{2}\right]$,
\item[$b)$]  $SN\left( Id+F -\dfrac{\epsilon}{2} \dfrac{\ket{m_i}\!\bra{m_i}}{\lambda_i^2}\right)=1$
, when  $k$ is arbitrary and $\frac{\epsilon}{2}\in\left [0,\frac{1}{12}\right]$, 
\item[$c)$]   $SN\left( Id+F -\dfrac{\epsilon}{2} \dfrac{\ket{m_i}\!\bra{m_i}}{\lambda_i^2}\right)=1$ , when $\mathbb{C}^k$ contains a SIC-POVM and $\frac{\epsilon}{2}\in \left [0,\frac{k+1}{12k}\right]$.
\end{itemize}

In addition, notice that $SR(\ket{v_{i_1,\ldots,i_n}})=n=\dfrac{SR(\ket{v})}{2}.$

\vspace{0.5cm}

So equation \ref{eqSN} provides a way to write $B$ using  only vectors with Schmidt rank less or equal to \\

\begin{itemize}
\item[$a)$]  $\max\left\{\dfrac{SR(\ket{v})}{2},2\right\}$, when  $k$ is arbitrary and $\epsilon\in\left [0,1\right]$,

\item[$b)$] $\dfrac{SR(\ket{v})}{2}$, when  $k$ is arbitrary and $\epsilon\in\left[0,\frac{1}{6}\right]$,

\item[$c)$] $\dfrac{SR(\ket{v})}{2}$, when $\mathbb{C}^k$ contains a SIC-POVM and  $\epsilon\in \left[0,\frac{k+1}{6k}\right]$.  
 \end{itemize}
 
  Hence,  $SN(B)\leq \max\left\{\dfrac{SR(\ket{v})}{2},2\right\}$ in case $a)$ and $SN(B)\leq \dfrac{SR(\ket{v})}{2}$ in cases $b)$ and $c)$.
 
 \vspace{0.5cm}
 
It remains to prove that $B$ is PPT for $\epsilon\in[0,1]$. 

It is not difficult to check that $\|\ket{v}\!\bra{v}^{\Gamma}\|_{\infty}=\max\{\lambda_1^2,\ldots, \lambda_n^2\}$. Since $ 2\left(\sum_{i=1}^n\lambda_i^2\right)=1$, we obtain  $$\|\ket{v}\!\bra{v}^{\Gamma}\|_{\infty}\leq \frac{1}{2}.$$

Finally, $\displaystyle (P_{sym}^{k,2})^{\Gamma}$ is positive definite and its minimum eigenvalue is  $\frac{1}{2}$. So $$B^{\Gamma}=(P_{sym}^{k,2})^{\Gamma}+\epsilon\ket{v}\!\bra{v}^{\Gamma}$$ 

is positive semidefinite for $\epsilon\in[0,1]$.
\end{proof}

\vspace{0.3cm}
  
  \begin{Theorem}\label{theoremPPT}Let $\gamma$ be a state supported on the antisymmetric subspace of  $\mathbb{C}^k\otimes \mathbb{C}^k$. 
Consider $B=P_{sym}^{k,2}+\epsilon \gamma\in M_k\otimes M_k$. Then   $B$ is PPT and
   
\begin{itemize}
\item[$a)$] 
$SN(B)=\left\{\begin{array}{ll}
\frac{SN(\gamma)}{2}, \text{ if } SN(\gamma)>2 \\
1 \text{ or } 2, \text{\  if } SN(\gamma)=2
\end{array}\right. $ for $\epsilon\in\left  ]0,1\right]$ and arbitrary k, \\

\item[$b)$]  $SN(B)=\frac{SN(\gamma)}{2}$  for  $\epsilon\in\left ]0,\frac{1}{6}\right]$ and arbitrary k, 
\item[$c)$]   $SN(B)=\frac{SN(\gamma)}{2}$  for  $\epsilon\in\left ]0,\frac{k+1}{6k}\right]$, if $\mathbb{C}^k$ contains a SIC-POVM .
\end{itemize}

  \end{Theorem}
\begin{proof} 
 First,  by  \cite[Proposition 1]{CarielloIneq}, $SN(B)\geq \dfrac{SN(\gamma)}{2}$ for every positive $\epsilon$.
  
Next,  let $\displaystyle \gamma=\sum_{i=1}^l\beta_i\ket{v_i}\!\bra{v_i}$, where $\sum_{i=1}^l\beta_i=1$,  $\beta_i>0$ and $\ket{v_i}$ is a unit antisymmetric vector of  $\mathbb{C}^k\otimes \mathbb{C}^k$ such that $SR(\ket{v_i})\leq SN(\gamma)$, for every $i$.\\
 
Finally, since  $B=P_{sym}^{k,2}+\epsilon \gamma =\sum_{i=1}^l\beta_i(P_{sym}^{k,2}+\epsilon   \ket{v_i}\!\bra{v_i})$, $$SN(B)\leq \max\left\{SN(P_{sym}^{k,2}+\epsilon   \ket{v_1}\!\bra{v_1}),\ldots, SN(P_{sym}^{k,2}+\epsilon   \ket{v_l}\!\bra{v_l}\right)\}.$$

 \vspace{0,3cm}

So the result follows by lemma \ref{lemmageneral}.
 \end{proof}

\vspace{0.3cm}

\section{Low operator Schmidt rank and separability}

 States of $\mathcal{M}_{3}\otimes \mathcal{M}_{m}$ with operator Schmidt rank 3 are in general not separable (\cite[Proposition 25]{Cariello_QIC}). Here we prove that invariance under left partial transpose is a sufficient condition for separability of such states.  This is a new result relating low operator Schmidt rank to separability (See \cite[Theorem 58]{cariello_QIC1}  and \cite[Theorem 19]{Cariello_QIC}).

  As a corollary we show that the Schmidt number of any state of $ \mathcal{M}_{k}\otimes \mathcal{M}_{m}\ (k\leq m)$ invariant under left partial transpose with operator Schmidt rank  3 cannot be greater than $ k-2$.  This result complements \cite[Theorem 5]{Marcus}. 

In this section,  let $\Ima(\delta)$ denote the image of $\delta\in \mathcal{M}_k\otimes \mathcal{M}_m$ within  $\mathbb{C}^k\otimes \mathbb{C}^m$.   \\

\noindent The next lemma is well known (e.g., \cite[Lemma 3.42]{Cariello_thesis} ). \\

\begin{lemma}\label{lemmashape} Any state $A\in \mathcal{M}_k\otimes \mathcal{M}_m$  with operator Schmidt rank $n$ can be written as  $A=\sum_{i=1}^n\gamma_i\otimes\delta_i$, where $\gamma_i\in \mathcal{M}_k,\delta_i\in \mathcal{M}_m$ are Hermitian matrices such that  $\Ima(\gamma_i)\subset\Ima(\gamma_1)$ and $\Ima(\delta_{i})\subset\Ima(\delta_1)$, for every $i$, and $\gamma_1,\delta_1$ are  positive semidefinite.\\
\end{lemma}

\begin{Theorem}\label{lowtensorrank}Let $A\in \mathcal{M}_{3}\otimes \mathcal{M}_{k}$ be a state which is invariant under left partial transpose. If its operator Schmidt rank is less or equal to 3 then $A$ is separable.
\end{Theorem}

\begin{proof}
 We can assume that the operator Schmidt rank of $A$ is 3, since every state with operator Schmidt rank less than 3 is separable by \cite[Theorem 58]{cariello_QIC1}.\\

First, let us assume that $A$ is positive definite. 
 Let  $A=\sum_{i=1}^3\gamma_i\otimes\delta_i$ be the decomposition described in lemma \ref{lemmashape}. 

Note that  $\gamma_1,\gamma_2,\gamma_3$ are real symmetric matrices,  since $A$ is invariant under left partial transpose.
Moreover,  $\gamma_1\in \mathcal{M}_3$ must be  positive definite, otherwise $A$ would not be positive definite (since $\Ima(\gamma_i)\subset\Ima(\gamma_1)$ for every $i$).\\

Let $\gamma_1=R^2$, where $R\in \mathcal{M}_3$ is real, symmetric and invertible. Let $B=(R^{-1}\otimes Id)A(R^{-1}\otimes Id)=$

\begin{center}
 $=Id_{3\times 3}\otimes \delta_1 + R^{-1}\gamma_2R^{-1}\otimes \delta_2+ R^{-1}\gamma_3R^{-1}\otimes \delta_3.$
\end{center}

Since $R^{-1}\gamma_2R^{-1}$ is real symmetric, there is an orthogonal matrix $O\in \mathcal{M}_3$ such that
$OR^{-1}\gamma_2R^{-1}O^t=D,$ where $D\in \mathcal{M}_3$ is a real diagonal matrix.\\

Let $C=(O\otimes Id)B(O^t\otimes Id)=Id_{3\times 3}\otimes \delta_1 + D\otimes \delta_2+ M\otimes \delta_3$, where $M\in \mathcal{M}_3$ is real symmetric.\\


Note that $C$ is positive definite
and  has the following format:

\begin{center}
$C=\left[\begin{array}{ccllrr}
F_1&m_{21}\delta_3&m_{31}\delta_3 \\
m_{21}\delta_3 &F_2&m_{32}\delta_3\\
m_{31}\delta_3 & m_{32}\delta_3&F_3\\
 \end{array}\right],$\\
\end{center}
 where $m_{ij}$ is the $ij$ entry of the real symmetric matrix $M$ and $\delta_3,F_1,F_2,F_3$ are Hermitian matrices. Since $C$ is positive definite,  $F_1\in \mathcal{M}_k$ is also positive definite.\\
  
Assume  that $m_{21},m_{31}\neq0$ (If one or both are equal to $0$ then the proof is simpler).
Note that   

 $$\left(\left[\begin{array}{ccllrr}
1&0&0 \\
0&1&0\\
0&m_{31}&-m_{21}\\
 \end{array}\right]\otimes Id_{k\times k}\right)C\left(\left[\begin{array}{ccllrr}
1&0&0 \\
0&1&m_{31}\\
0&0&-m_{21}\\
 \end{array}\right]\otimes Id_{k\times k}\right)=\left[\begin{array}{ccllrr}
F_1& m_{21}\delta_3&0\\
m_{21}\delta_3 &H_{2}&H_{3}\\
0&H_{3}&H_{4}\\
 \end{array}\right].$$

\vspace{0,5cm}

Next, let  $F_1=UU^*$ for an invertible $U$. Thus, 

 $$(Id_{3\times 3}\otimes U^{-1})\left[\begin{array}{ccllrr}
F_1&m_{21}\delta_3&0\\
m_{21}\delta_3 &H_{2}&H_{3}\\
0&H_{3}&H_{4}\\
 \end{array}\right] (Id_{3\times 3}\otimes U^{-1})^*=\left[\begin{array}{ccllrr}
Id_{k\times k}&L&0 \\
L &O_{2}&O_{3}\\
0&O_{3}&O_{4}\\
 \end{array}\right].$$

\vspace{0,5cm}

Note that $L$ is Hermitian, since $L=U^{-1}(m_{21}\delta_3)(U^{-1})^* $.

 Now,

\begin{equation}\label{eqC}
 \left[\begin{array}{ccllrr}
Id_{k\times k}&L&0 \\
L &O_{2}&O_{3}\\
0&O_{3}&O_{4}\\
 \end{array}\right]=\left[\begin{array}{ccllrr}
0&0&0 \\
0&O_{2}-L^2&O_{3}\\
0&O_{3}&O_{4}\\
 \end{array}\right]+\left[\begin{array}{ccllrr}
Id_{k\times k}&L&0\\
L&L^2&0\\
0&0&0\\
 \end{array}\right]
\end{equation}

\vspace{0,5cm}
 
 The second summand above is a well known separable matrix, since $L$ is Hermitian (See \cite[Theorem 1]{Kraus} and \cite[Lemma 3]{Horodecki}). \\

In addition, the first summand can be embedded in $\mathcal{M}_2\otimes \mathcal{M}_k$. Since there are only three  sub-blocks  forming this matrix ($O_{2}-L^2, O_3$ and $O_4$), its operator Schmidt rank is less or equal to 3. Moreover, it is positive semidefinite, since

$$\left[\begin{array}{ccllrr}
0&0&0 \\
-L &Id&0\\
0&0&Id\\
 \end{array}\right]
\left[\begin{array}{ccllrr}
Id &L&0 \\
L&O_{2}&O_{3}\\
0&O_{3}&O_{4}\\
 \end{array}\right] \left[\begin{array}{ccllrr}
0&-L&0\\
0&Id&0\\
0&0&Id\\
 \end{array}\right]=
 \left[\begin{array}{ccllrr}
0&0&0 \\
0&O_{2}-L^2&O_{3}\\
0&O_{3}&O_{4}\\
 \end{array}\right].$$
 
\vspace{0,5cm}

 Therefore, the first summand of equation \ref{eqC} is also separable by \cite[Theorem 19]{Cariello_QIC}. Hence, the sum is separable. Since all the local operations used are reversible and preserve separability,  $A$ is separable. 
 
 Now, for the positive semidefinite case.  Given $\epsilon>0$, define $A(\epsilon)=(\gamma_1+\epsilon Id)\otimes(\delta_1+\epsilon Id)+\gamma_2\otimes\delta_2+\gamma_3\otimes\delta_3.$

Note that $A(\epsilon)$ has operator Schmidt rank less or equal to 3,  is invariant under left partial transpose $(\epsilon Id+\gamma_1,\gamma_2,\gamma_3$ are  symmetric$)$ and  
 is  positive definite $(A(\epsilon)=A+\epsilon Id\otimes \delta_1+\gamma_1\otimes \epsilon Id+\epsilon^2 Id\otimes Id)$.  By the first case, $A(\epsilon)$ is separable and so is $\displaystyle\lim_{\epsilon\rightarrow 0+} A(\epsilon)=A$. \end{proof} 
 
 \vspace{0,5cm}

\begin{corollary}\label{corollaryinvariant} Let $A\in \mathcal{M}_{k}\otimes \mathcal{M}_{m}\ (k\leq m)$ be a positive semidefinite Hermitian matrix which is  invariant under left partial transpose. If its operator Schmidt rank is  equal to 3 then $SN(A)\leq k-2$.
\end{corollary}
\begin{proof}Let us show that $SN(A)$ cannot be $k-1$, since $SN(A)<k$ was already proved in  \cite[Theorem 5]{Marcus}.

 If $SN(A)=k-1$ then $A$ contains  an entangled sub-block, $B\in \mathcal{M}_3\otimes \mathcal{M}_{m}$, which is invariant under left partial transpose (See \cite[Theorem 4]{Marcus} and \cite[Theorem 5]{Marcus}  for details). 
 
By the construction of $B$, its operator Schmidt rank is less or equal to the operator Schmidt rank of $A$, which is 3. Therefore, $B$ is separable by  theorem \ref{lowtensorrank}.  Absurd!
\end{proof}

  \section*{Summary and Conclusion}

In this work, we obtained results on the number of mutually unbiased bases and the Schmidt number of states under certain constraints. The connection between these different results is the type of restrictions imposed. These restrictions were made on the Schmidt rank of the tensors used in the results.

We obtained an upper bound on the number of mutually unbiased bases  of $\mathbb{C}^m\otimes\mathbb{C}^n$ formed by vectors with Schmidt rank less or equal to $k$ $(k<m\leq n)$. It solved a conjecture on mutually unbiased product bases in a straightforward way. 

We found an interval for the values of $\epsilon$ such that the Schmidt number of $P_{sym}^{k,2}+\epsilon\gamma$ equals half of the Schmidt number of $\gamma$ for all states $ \gamma $  supported on the antisymmetric subspace of $\mathbb{C}^k\otimes \mathbb{C}^k$. This common interval provided a flexible method to create PPT states with high Schmidt numbers.
 
Finally,  we proved that  invariance under left partial transpose  is a sufficient condition for the separability of operator Schmidt rank 3 states of $\mathcal{M}_3\otimes \mathcal{M}_m$. As a corollary we proved that the Schmidt number of operator Schmidt rank 3 states of $\mathcal{M}_k\otimes \mathcal{M}_m\ (k\leq m)$ that are invariant under left partial transpose  cannot exceed $k-2$. 

\section*{Acknowledgment} The author would like to thank the referee for providing constructive comments and helping in the improvement of this manuscript. 
  
  \section*{Disclosure Statement}
  No potential conflict of interest was reported by the author.

\end{document}